\documentclass[final,12pt]{article}
\usepackage{hyperref}
\usepackage{upgreek}
\usepackage{textgreek}
\usepackage{amssymb,amsmath,amsthm}
\usepackage{cite}
\usepackage{tikz}
\usepackage[utf8]{inputenc}
\usetikzlibrary{arrows,backgrounds,decorations.pathmorphing,decorations.pathreplacing,positioning,fit}
\usepackage{enumerate}
\usepackage{multicol}
\usepackage{accents}
\usepackage{stackengine}
\usepackage[conditional,light,first,bottomafter]{draftcopy}
\draftcopyName{DRAFT\space\today}{130}
\draftcopySetScale{65}
\usepackage{lmodern}
\usepackage{microtype}
\usepackage{graphicx}
\usepackage{subcaption}
\usepackage{dsfont}
\usepackage[letterpaper,hmargin=2.5cm,vmargin=2.8cm]{geometry}
\usepackage{setspace}
\usepackage{colonequals}
\makeatletter
\renewcommand{\section}{\@startsection%
{section}%
{1}%
{0em}%
{1.7em}%
{1.2em}%
{\normalfont\large\centering\bfseries}}
\renewcommand{\@seccntformat}[1]%
{\csname the#1\endcsname.\hspace{0.5em}}
\makeatother

\numberwithin{equation}{section}
\newtheorem{theorem}{Theorem}[section]
\newtheorem{proposition}[theorem]{Proposition}
\newtheorem{lemma}[theorem]{Lemma}
\newtheorem{corollary}[theorem]{Corollary}
\theoremstyle{definition}
\newtheorem{definition}[theorem]{Definition}
\newtheorem{remark}[theorem]{Remark}

\newcommand{\abs}[1]{\left|#1\right|}
\newcommand\rE[1]{\upharpoonleft_{#1}}
\newcommand\lrb[1]{\left\lbrace#1 \right\rbrace}
\newcommand\lrp[1]{\left(#1 \right)}
\newcommand\h{{\initsp}}
\newcommand\C{{\mathbb C}}
\newcommand\N{{\mathbb N}}
\newcommand\R{{\mathbb R}}
\newcommand\cc[1]{\overline {#1}}
\newcommand\vb[1]{\!\left\langle {#1} \right |}
\newcommand\vk[1]{\left |{#1}\right\rangle\!}
\newcommand\ip[2]{\left\langle {#1},{#2} \right\rangle}
\newcommand\no[1]{\left\| {#1} \right\|}
\newcommand{\initsp}{\mathsf{h}}
\newcommand\unit{\hbox{\rm 1\kern-2.8truept l}}
\newcommand\qF[1]{\mathcal F[#1]}
\newcommand\tr[1]{{{\rm tr}}\left(#1\right)}
\newcommand\wm[1]{\left\langle {#1}\right\rangle}

\DeclareMathOperator{\im}{Im}
\DeclareMathOperator{\re}{Re}
\DeclareMathOperator{\dom}{dom}
\DeclareMathOperator{\ran}{ran}
\DeclareMathOperator{\Span}{span}

\makeatletter
\newcommand{\mathleft}{\@fleqntrue\@mathmargin0pt}
\newcommand{\mathcenter}{\@fleqnfalse}
\makeatother
\begin{document}
\begin{titlepage}
\title{On the analytical approach to infinite-mode Boson-Gaussian states
\footnotetext{%
Mathematics Subject Classification(2010):
Primary 81S05; 
Secondary 78M05, 
60B15. 
}
\footnotetext{%
Keywords: 
Analytic quantum Gaussian state, non-commutative Fourier transform, Weyl  moments
}
}
\author{
\textbf{Jorge R. Bola\~nos-Serv\'in, Roberto Quezada and Josu\'e I. Rios-Cangas}
\\
\small Departamento de Matem\'aticas\\[-1.6mm] 
\small Universidad Aut\'onoma Metropolitana, Iztapalapa Campus\\[-1.6mm]
\small San Rafael Atlixco 186, C.P. 09340, Iztapalapa, Mexico City.\\[-1.6mm] 
\small \texttt{jrbs@xanum.uam.mx,\quad roqb@xanum.uam.mx,\quad jottsmok@xanum.uam.mx}
}
\date{}
\maketitle
\vspace{-4mm}
\begin{center}
\begin{minipage}{5in}
  \centerline{{\bf Abstract}} \bigskip
We develop an analytical approach to quantum Gaussian states in infinite-mode representation of the Canonical Commutation Relations (CCR's), using Yosida approximations to define integrability of possibly unbounded observables with respect to a state $\rho$ ($\rho$-integrability). It turns out that all elements of the commutative $*$-algebra generated by a possibly unbounded $\rho$-integrable observable $A$, denoted by $\langle A\rangle$, are normal and $\rho \, $-integrable. Besides, $\langle A\rangle$ can be endowed  with the well-defined norm $\|\cdot\|_\rho:= \tr{\rho |\cdot| }$. Our approach allows us to rigorously establish fundamental properties and derive key formulae for the mean value vector and the covariance operator. We additionally show that the covariance operator $S$ of any Gaussian state is real, bounded, positive, and invertible, with the property that $S-iJ\geq 0$, being $J$ the multiplication operator by $-i$ on $\ell_2({\mathbb N})$.  
\end{minipage}
\end{center}
\thispagestyle{empty}
\end{titlepage}

\section{Introduction} 
Quantum Gaussian states are a natural extension to the quantum setting of the notion of Gaussian or normal distributions in classical probability. A pedagogical introduction based on Bochner's theorem on characteristic functions, addressed to classical probabilistic, can be found in Parthasarthy's paper \cite{MR2662722}. Nowadays, such class of quantum states together with quantum Gaussian channels and semigroups are intensively used in Quantum Physics, for instance, to model noise in electromagnetic communications. Due to their importance in quantum optics and quantum information theory, see \cite{MR4530560,MR2304059,MR2513281,Caruso_2006,Caruso_2008,Kruger06,MR3190189,Zhang18,Hackl18EB} and the references therein; the need for a rigorous mathematical foundation of their analytical structure  has  become a priority and thus has induced researchers to actively work on deriving mathematical frameworks where gaussianity properties are better understood \cite{MR2964163,MR4537381,MR4133450,MR4158636,MR4440218,MR2662722,MR4218302,MR4014466,MR3666198}.  \\ 
While most previous research, including Parthasarathy's approach, focuses on the finite-mode scenario, the infinite-mode case presents a fascinating subject not only from a physical perspective but also due to the emergence of additional mathematical complexities. In \cite{MR4014466} the authors developed further Parthasarathy's approach in an infinite-mode scheme using Bogoliubov transformations, Shale operators and an extension of Williamson's theorem to infinite dimension, which leads to a characterization of the infinite-mode covariance operator. A more direct approach however, starting with Parthasarathy's definition and using Yosida approximations to rigorously define the integrability of unbounded observables with respect to a state, allowed us to extend Parthasarathy's analytical approach to the infinite-mode case, rigorously establishing fundamental properties and deriving key formulae for the mean value vector and the covariance operator, which are often cited  in the physics literature without a proof. We emphasize the analytical nature of Parthasarthy's approach in contradistinction to the combinatorial approach to gaussianity based on the characterization of normal random variables with all moments, in terms of their second order moments. In the one-mode case, this analytical approach was rigorously detailed in \cite{MR4516426}. This paper seeks to extend these results to the infinite-mode case.    

 The paper layout is as follows. In Section \ref{Weyl-representation-CCR} we briefly discuss the infinite-mode Weyl representation of the CCR's. We introduce the notions of analytic quantum boson Gaussian state and integrability of possible unbounded observables with respect to a state in Section \ref{Analytic-Gaussian-States}. A criterion for integrability of possibly unbounded normal operators with respect to a state $\rho$ is given in Theorem \ref{thm:integrable_norm}, it turns out that all elements in the commutative $*$-algebra generated by a $\rho\,$-integrable normal operator $A$, denoted by $\langle A\rangle$, are $\rho\,$-integrable and normal and $\|\cdot\|_\rho= \tr{\rho |\cdot| }$ is a norm on $\langle A\rangle$. Up to our knowledge, all these results are new. 
 
Finally in Section \ref{Covariance-operator}, under the assumption that the normal operator $p(z)p(u)$ is $\rho\,$-integrable for all $z, u\in\ell_2({\mathbb N})$, where $p(z)$ is the field operator, we deduce mathematically rigorous formulae for the matrix elements of the covariance operator and the mean value vector of Gaussian states. We prove  that for any amenable Gaussian state the covariance operator is real and bounded, i.e. $S\in{\mathcal B}_{\mathbb R}(\ell_2({\mathbb N}))$, positive and invertible with 
\[0\leq \|S^{-1}\|\leq 1\leq \|S\|<\infty\,.\] Moreover, $S-iJ\geq 0$, where $J$ is the multiplication operator by $-i$ on $\ell_2({\mathbb N})$ (see Corollary~\ref{cor:uncertainty-principle}).  Although these results may not be entirely novel, our approach sheds some new light on the mean value vector and covariance operator properties of a boson Gaussian state without relying on  Shale operators or extensions of Williamson's theorem.

\section{Infinite-mode Weyl representation of the CCR's}\label{Weyl-representation-CCR}
 We start by recalling the definition of stabilized infinite tensor product of Hilbert spaces. \begin{definition}
 Given a family $\{\initsp_j\}_{j\geq 1}$ of separable Hilbert spaces and a sequence of unit vectors $\varphi=(\varphi_j)_{j\geq 1}$ with $\varphi_j\in \initsp_j$, the $\varphi$-stabilized infinite tensor product 
 \begin{gather}\label{eq:stabilized-TP}
 \bigotimes_{j\geq 1}\lrp{\initsp_j, \varphi_j}
 \end{gather}
  is defined as the completion of the linear span of 
\begin{gather*}
 F=\{(v_j)_{j\geq1} \, : \; v_j\in \initsp_j, \; \forall j\geq 1 \; \textrm{and} \; v_j= \varphi_j \; \textrm{for all but finitely many} \; j\geq 1\} \nonumber
\end{gather*} 
with respect to the norm induced by the inner product 
\begin{align}\label{stab-inner-product}
 \big\langle (v_1, v_2, \cdots ), (w_1, w_2, \cdots)\big\rangle=\prod_{j\geq 1}\langle v_j, w_j \rangle 
\end{align}
 \end{definition}
 
For any $v=(v_1, \cdots, v_n, \varphi_{n+1}, \cdots),\, w=(w_1, \cdots, w_m, \varphi_{m+1}, \cdots)\in F$, the product in \eqref{stab-inner-product} is finite since by Schwarz inequality,  
\begin{gather*}
\abs{\ip vw}=\prod_{j\geq 1}\abs{\ip{v_j}{w_j}} \leq \prod_{j\geq 1}\no{v_j}\no{w_j}= \prod_{j=1}^n \no{v_j}\prod_{j=1}^m \no{w_j}\,.
\end{gather*}
 
 An element $v$ in $\bigotimes_{j\geq 1}(\initsp_j, \varphi_j)$ has the form
 \begin{align}\label{typical-element}
 v=\lim_n \lrp{v_1\otimes \cdots \otimes v_n \otimes \varphi_{n+1} \otimes \cdots}\,,
 \end{align}
where the limit is in the norm induced by the inner product \eqref{stab-inner-product}. But, assuming $\|v_j\|=1, \, j=1, \cdots, m$, we have for $m>n$, 
\begin{align*}
 &\no{v_1 \otimes  \cdots \otimes v_n \otimes \varphi_{n+1} \otimes \cdots - v_1 \otimes  \cdots \otimes v_m \otimes \varphi_{m+1}  \otimes \cdots}^2
 \\ &=\,\no{v_1\otimes \cdots \otimes v_n \otimes (\varphi_{n+1} \otimes \cdots \otimes \varphi_m- v_{n+1}\otimes\cdots \otimes v_m)\otimes \varphi_{m+1}\otimes\cdots}^2 \\ 
 &=\, 2 - \prod_{j=n+1}^m \ip{\varphi_j}{ v_j}-\prod_{j=n+1}^m\ip{ v_j}{ \varphi_j}\,.
\end{align*}
Thus, to ensure the convergence in \eqref{typical-element} we assume that  
\begin{gather}\label{eq:convergence-condition}
 \sum_{j\geq 1}\abs{\ip{v_j}{\varphi_j}-1}<\infty
\end{gather}
holds, viz. $v_j$ and $\varphi_j$ are close enough to each other for large values of $j$. 
Now, taking $\initsp_j =L_2({\mathbb R})$ and the  vacuum vector $\varphi_j=\unit\in L_2({\mathbb R})$, with $\unit(x)=\pi^{-\frac{1}{4}} e^{-\frac{x^2}{2}}$, for each $j\geq 1$, we denote the stabilized infinite tensor product \eqref{eq:stabilized-TP} by 
\begin{gather}\label{eq:infinite-TP}
\h\colonequals  \bigotimes_{j\geq 1}\big(L_2({\mathbb R}), \unit\big)\,.
\end{gather} or simply by $\initsp=\bigotimes_{j\geq 1}L_2({\mathbb R})$. 

Along the paper, $\Gamma(\ell_2({\mathbb N}))$ denotes the symmetric Fock space constructed on the Hilbert space $\ell_2({\mathbb N})$. There is a natural identification of $\Gamma(\ell_2({\mathbb N}))$ with $\initsp$, such that the corresponding exponential vector $\varepsilon_z\in \Gamma\lrp{\ell_2({\mathbb N})}$ satisfies 
\begin{gather}
\label{infty-exp-vectors}
\varepsilon_z = \otimes_{j\geq 1} \varepsilon_{z_j} := \lim_{n}\lrp{\otimes_{j=1}^n \varepsilon_{z_j} \otimes \otimes_{j\geq n+1}\unit}\,,\quad z=(z_j)_{j\geq 1}\in\ell_2(\mathbb{N})
\end{gather}
where $\varepsilon_{z_j}(x)= e^{\sqrt{2}z_j x-z_j^2/2} \unit(x)$ are the exponential vectors in $L_2(\mathbb{R})$. 
Note that 
\begin{gather*}
\langle \varepsilon_{z_j}, \unit\rangle=\langle \varepsilon_{z_j}, \varepsilon_0\rangle = e^{\langle z_j, 0\rangle}= 1\qquad\mbox{and}\qquad \sum_{j\geq 1}\abs{\ip{\varepsilon_{z_j}}{\unit}-1}=0\,.
\end{gather*}
Thus, one has by \eqref{eq:convergence-condition} that the limit in \eqref{infty-exp-vectors} exists. The density of the exponential domain 
\begin{gather*}
\mathcal{E}=\Span\lrb{\varepsilon_z \, : \, z\in\ell_2(\mathbb{N})} 
\end{gather*}
follows from the density of the corresponding exponential domain in each coordinate. Indeed, given any 
\begin{gather*}
v=\lim_{n}\lrp{\bigotimes_{j=1}^n v_{j} \otimes \bigotimes_{j\geq n+1}\unit}\in \h\,,
\end{gather*}
for $0<\epsilon<1$ and each $1\leq j\leq n$, there exists a linear combination of exponential vectors for which
\begin{gather*}
\no{v_j - \sum_{k=1}^{n_j} \lambda_{jk} \varepsilon_{z_{jk}}} <\frac{\epsilon}{2}\,.
\end{gather*}
So, if $\sigma_n = (\otimes_{j=1}^{n} \sum_{k=1}^{n_j} \lambda_{jk} \varepsilon_{z_{jk}})\otimes \otimes_{j\geq n+1}\unit \in \mathcal{E}$ and $v_n =\otimes_{j=1}^n v_j \otimes \otimes_{j\geq n+1}\unit$, then
\begin{gather*}
\no{v_n-\sigma_n}^2 = \prod_{j=1}^n \no{v_j - \sum_{k=1}^{n_j} \lambda_{jk} \varepsilon_{z_{jk}}}^2 < \lrp{\frac{\epsilon}{2}}^{2n}\,,
\end{gather*}
which taking $n$ large enough so that $\no{v-v_n}<\epsilon/2$, we get     
\begin{gather*}
\no{v-\sigma_n}= \no{v-v_n} + \no{v_n -\sigma_n}< \frac{\epsilon}{2} + \lrp{\frac{\epsilon}{2}}^n <\epsilon\,. 
\end{gather*}

The correspondence  
\begin{gather*}
\varepsilon_u \mapsto e^{-\frac12 \no{z}^2- \ip zu} \varepsilon_{z+u}\,,\qquad z\in\ell_2(\mathbb{N})\,,
\end{gather*}
 is an isometry from $\mathcal{E}$ onto itself. Indeed, the identity
\begin{align*}
\ip{\varepsilon_u}{\varepsilon_v}= &e^{\ip uv}= e^{-\no{z}^2} e^{-\ip zv-\overline{\ip{z}{ u}}}e^{\ip{u+z}{v+z}} \\=&\ip{ e^{-\frac12 \no z^2- \ip zu}\varepsilon_{z+u}}{e^{-\frac12 \no z^2- \ip zv} \varepsilon_{z+v}}
\end{align*}
extends to $\mathcal{E}$ by linearity. Thereby, there exists a unique unitary operator $W_z$ on the space $\Gamma(\ell_2({\mathbb N}))\simeq \bigotimes_{j\geq 1}L_2({\mathbb R})$, satisfying (cf. \cite[Prop.\,7.2]{MR3012668})
\begin{gather}\label{eq:Weyl-displacement}
W_z \varepsilon_u= e^{-\frac12 \no{z}^2- \ip zu} \varepsilon_{z+u}\,,
\end{gather}
which is called the \emph{Weyl operator} (or displacement operator) associated with $z\in\ell_2(\mathbb{N})$. Moreover,  
\begin{align*}
\langle W_z\varepsilon_u, \varepsilon_v\rangle&= e^{-\frac12 \|z\|^2- \langle u, z\rangle} \langle \varepsilon_{z+u}, \varepsilon_v\rangle= e^{-\frac12 \|z\|^2- \langle u, z\rangle} e^{\langle z+u, v\rangle} \\ &= e^{\langle u, v-z\rangle} e^{-\frac12 \|z\|^2- \langle -z, v\rangle}  = \langle \varepsilon_u, e^{-\frac12 \|z\|^2- \langle -z, v\rangle} \varepsilon_{-z +v}\rangle \\&= \langle\varepsilon_u, W_{-z}\varepsilon_v\rangle\,.
\end{align*}
Hence, we conclude by linearity and density that $W_z^*=W_{-z}$. Also, the infinite-modes CCR in Weyl form 
\begin{gather}\label{infty-CCR-weyl}
W_zW_u=e^{-i\im\ip zu}W_{z+u}\,,\qquad W_zW_u=e^{-2i\im\ip zu}W_uW_z\,,\quad z,u\in\ell_2(\mathbb{N})
\end{gather} hold true. Indeed, 
\begin{align*}
W_z W_u\varepsilon_v &= e^{-\frac12 \|u\|^2- \langle u, v\rangle} W_z \varepsilon_{u+v} =  
e^{-\frac12 \|u\|^2- \langle u, v\rangle} e^{-\frac12 \|z\|^2- \langle z,u+ v\rangle}\varepsilon_{z+u+v} \\&= e^{-\frac12 \|z+u\|^2 - \langle z+u, v\rangle} e^{-i\im\ip zu}\varepsilon_{z+u+v} = e^{-i\im\ip zu} W_{z+u}\varepsilon_v\,,
\end{align*}
while
\begin{align*}
W_z W_u\varepsilon_v &=  e^{-i\im\ip zu} W_{z+u}\varepsilon_v= e^{-i\im\ip zu}e^{i\im\ip uz} W_u W_z \varepsilon_v \\&= e^{-2i\im\ip zu} W_u W_z \varepsilon_v\,.
\end{align*}
 
For $t\in\mathbb{R}$ and $z\in\ell_2(\mathbb{N})$, we have  
\begin{gather*}
W_{tz}W_{sz}=e^{-its\im\ip zz}W_{tz+sz} = W_{(t+s)z}\,,
\end{gather*} 
which proves that $t\mapsto W_{tz}$ is a unitary group for each fixed $z$ and the unitary group $(W_{tz})_{t\in\mathbb{R}}$ is strongly continuous. Let us denote by $p(z)$ the infinitesimal generator of $(W_{tz})_{t\in\mathbb{R}}$, so that
\begin{gather}\label{eq:Weyl-generator}
 W_{tz}=e^{-itp(z)}\,.
\end{gather}
The generator $p(z)$ is called the \emph{infinite-mode field} operator. For each $z\in\ell_2(\mathbb{N})$, $p(z)$ is unbounded and self-adjoint, with a densely defined domain, which contains $\mathcal{E}\subset\h$. Besides, the set of all finite particle vectors is a core for $p(z)$  \cite[Sect.\,20]{MR3012668}. Moreover, 
\begin{gather}\label{commutation-sigma}
[p(z), p(u)]\varepsilon_w= 2i\im \ip zu \varepsilon_w\,,\quad z,u, w\in\ell_2(\mathbb{N})\,,
\end{gather}
whence if $p(z)=0$ then $z=0$.

Let $\{\delta_j\}_{j \geq 1}$ be the canonical basis of $\ell_2(\mathbb{N})$ and define 
\begin{align}\label{eq:MP-CA}
\begin{split}
p_j&\colonequals 2^{-\frac{1}{2}}p(\delta_j)\,,\qquad\qquad q_j\colonequals -2^{-\frac{1}{2}}p(i \delta_j)\,, \\ 
a_j &\colonequals 2^{-\frac{1}{2}} (q_j + i p_j)\,,\qquad a_j^{\dagger}\colonequals 2^{-\frac{1}{2}}(q_j -i p_j)\,.
\end{split}
\end{align}
The operators $p_j,q_j$ are selfadjoint, while $a_j,a_j^\dagger$ are closed and satisfy $a_j^*=a_j^\dagger$. Besides, it readily follows from \eqref{commutation-sigma} that 
\begin{gather*}
[q_j, p_k] =  i \delta_{jk} I\,,\quad [p_j, p_k] = 0\,,\quad [q_j, q_k] =0\,,\quad (j,k\geq1)
\end{gather*}
which implies $[a_j, a_{k}^\dagger]= \delta_{jk}I$.  The elements of the families $\{a_j\}_{j\geq1}$ and $\{a_j^\dagger\}_{j\geq 1}$ commute among themselves. The operators $q_j, p_j$ are called the \emph{position} and \emph{momentum}, respectively, while $a_j, a_j^\dagger$ are the \emph{annihilation} and \emph{creation} operators, respectively. 

For $z=\sum_{j=1}^n(x_j+iy_j)\delta_i\in \ell_2(\mathbb N)$, with $\{x_j,y_j\}_{j=1}^n\subset \R$, one has
\begin{gather*}
e^{-itp(z)}=W_{tz}=e^{-it\sqrt{2}\sum_{j=1}^n(x_jp_j-y_jq_j)}\,,\qquad p(z)=\sqrt{2}\sum_{j=1}^n(x_jp_j-y_jq_j)\,.
\end{gather*}
Besides, for $z,u\in \ell_2(\N)$, the left-hand side of \eqref{infty-CCR-weyl} implies $e^{-it^2\im\ip zu}W_{tz + tu}= W_{tz} W_{t u}$, which by derivation at $t=0$ one obtains 
\begin{gather}\label{eq:linear-P-inZ}
p(z+u)=p(z)+p(u)\,.
\end{gather}

\section{Analytic approach to boson Gaussian states}\label{Analytic-Gaussian-States}
Let $L_{1}(\initsp)$ and $L_{2}(\initsp)$ be the Banach and Hilbert spaces of finite trace and Hilbert-Schmidt operators on $\initsp$, with norm $\no{\rho}_1=\tr{\abs{\rho}}$ and inner product $\ip{\rho}{\eta}_2 = \tr{\rho^* \eta}$, respectively.  
We remark that $A-\alpha$ means $A-\alpha I$, for any operator $A$ and $\alpha\in\C$.

\begin{definition}
The \emph{Wigner or non-commutative Fourier} transform of $\rho\in L_{1}\lrp{\initsp}$, is   
\begin{gather*}
\qF\rho(z)\colonequals\frac{1}{\sqrt{\pi}} \tr{\rho W_z}, \quad z\in\ell_2(\mathbb{N})
\end{gather*}
which is a complex-valued function on $\ell_2(\mathbb{N})$. 
\end{definition}

Consider the real Hilbert space $(\ell_2(\mathbb{N}),(\cdot,\cdot))$, where $(z, u)\colonequals\re\ip zu$, with canonical basis $\{\delta_j,-i\delta_j\}_{j\geq1}$, being $\{\delta_j\}_{j\geq1}$ the canonical basis of $\ell_2(\mathbb{N})$.

\begin{definition}\label{rm:properties-wS}
We call a state $\rho \in L_1(\initsp)$  \emph{analytic boson Gaussian} (or Gaussian for short), if there exist $w\in\ell_2(\mathbb{N})$ and a selfadjoint real operator $S\in{\mathcal B}_{\mathbb R}({\ell_2(\mathbb{N})})$, for which 
\begin{gather}\label{eq:properties-wS}
\qF\rho(z)=\frac{1}{\sqrt{\pi}}e^{-i(w,z)-\frac12 (z,Sz)}\,,\quad \forall z\in\ell_2(\mathbb{N})\,.
\end{gather}
In such a case, we write $\rho=\rho(w,S)$.
\end{definition} 

The elements $w $ and $S$ of a Gaussian state are called the mean value vector and the covariance operator, respectively. The applications  $z \mapsto (w, z)$ and $z \mapsto (z, Sz)$, determine a real linear functional and a real quadratic form, respectively, on the real Hilbert space $\ell_2(\mathbb{N})$. Thereby, the components $w, S$ are uniquely determined by the definition. 

\begin{remark}[Coherent channel]\label{rm:coherent-chanel} For a Gaussian state $\rho(w,S)$, it follows that 
\begin{gather}\label{eq:coherent-chanel}
T_u(\rho)=W_u\rho W_u^*\,,\quad u\in\ell_2(\N)
\end{gather}
is a Gaussian state with mean value vector $w-2iu$ and the covariance operator $S$. Indeed, $-2i\im\ip zu=-i(-2iu,z)$ and the right-hand side of \eqref{infty-CCR-weyl} implies 
\begin{align*}
\qF{T_u(\rho)}(z)=e^{-i(-2iu,z)}\qF{\rho}(z)=\frac{1}{\sqrt{\pi}}e^{-i(w-2iu,z)-\frac12 (z,Sz)}\,,
\end{align*}
as expected. The map \eqref{eq:coherent-chanel} is called the coherent channel (cf. \cite[Ex.\,3]{MR4516426}).
\end{remark}

We will follow the approach in \cite{MR4516426}, to define moments of an unbounded observable in a state $\rho$, using well-known properties of Yosida's approximations, for unbounded selfadjoint operators \cite[Sec.\,1.3]{MR710486}. Given $\epsilon>0$ and an infinitesimal generator $\Lambda$ of a strongly continuous semigroup of contractions $e^{-t\Lambda}$, with $t\geq0$, it follows that  $(I+\epsilon \Lambda)^{-1}$ is a contraction in $\mathcal B(\mathcal{H})$, and $\Lambda_{\epsilon}\colonequals \Lambda(I + \epsilon \Lambda)^{-1}\in\mathcal B(\mathcal{H})$. In addition, for $u\in \mathcal H$,
\begin{gather*}
 \lim_{\epsilon\to 0} \lrp{I+\epsilon \Lambda}^{-1}u=u\,,\quad \lim_{\epsilon\to 0} e^{-t\Lambda_{\epsilon}}u = e^{-t\Lambda}u\,,
\end{gather*}
while for $u\in \rm{dom \, } \Lambda$,
\begin{gather*}
(I+\epsilon \Lambda)^{-1}\Lambda u=\Lambda(I+\epsilon \Lambda)^{-1}u\,,\quad\lim_{\epsilon\to 0} \Lambda_{\epsilon}u = \Lambda u\,.
\end{gather*}
The operator $\Lambda_{\epsilon}$ is the infinitesimal generator of the uniformly continuous semigroup of contractions $e^{-t\Lambda_{\epsilon}}$. The above reasoning is satisfied for $\Lambda=\pm iA$, where $A$ is an observable, i.e., a selfadjoint operator. 

\begin{definition}\label{def-Weyl-moments} Given a state $\rho$ and $n\in\mathbb N$, the \emph{$n^{\rm th}$ moment} of an unbounded observable $A$ at $\rho$ (briefly the $\rho\,$-$n^{\rm th}$ moment of $A$) is defined by 
\begin{align}\label{moments-derivatives}
\wm{A^n}_\rho\colonequals (-i)^n \lim_{\epsilon\to 0}\tr{\rho\big((iA)_{\epsilon}\big)^n} = (-i)^n {\frac{d^n}{dt^n}\textrm{tr}\big(\rho e^{it A}\big)}\upharpoonleft_{t=0}\,,
\end{align}
whenever the limit exists \cite[Th.\,1]{MR4516426}. In the case when $A$ is bounded, $\tr{\rho A^n}$ is recovered. 
\end{definition}

For a Gaussian state $\rho(w,S)$, the $\rho(w,S)$-$n^{\rm th}$ moments of the field operator $p(z)$ are 
\begin{gather*}
\wm{p(z)^n}_\rho\colonequals (-i)^n \lim_{\epsilon\to 0}\tr{\rho\lrp{(ip(z))_{\epsilon}}^n} = (-i)^n {\frac{d^n}{dt^n}\tr{\rho e^{- it p(z)}}}\upharpoonleft_{t=0}\,,\quad z\in\ell_2(\mathbb{N})\,,
\end{gather*}
which are called the \emph{$\rho\,$-$n^{\rm th}$ Weyl moments}. 
\begin{remark}  All the $n^{\rm th}$-Weyl moments of $\rho(w,S)$ are finite. Indeed, from   \eqref{moments-derivatives} and \eqref{eq:properties-wS} it follows that 
\begin{gather*}
\wm{p(z)^n}_\rho= (-i)^n {\frac{d^n}{dt^n}e^{-it(w,z)-\frac{1}{2}t^2(z,Sz)}}\rE{t=0}\,.
\end{gather*} The above quality suggests that Gaussian states belong to a class of docile or tractable states in the sense that allows explicit measurements of significant physical quantities such as energy, position and moment.
\end{remark}
The characteristic function of the field operator $p(z)$ in a state $\rho$ is
\begin{align}\label{charact-Gaussian}
\varphi_z[\rho](t)\colonequals\mathcal F[\rho](tz)=\frac{1}{\sqrt{\pi}}\tr{\rho W_{tz}}\,, \quad t\in\mathbb R\,.\end{align}
Thus, for a  Gaussian state  $\rho(w,S)$ one has 
 $\varphi_z[\rho](t)=\pi^{-1/2}e^{-it(w,z)-\frac{1}{2}t^2 (z,Sz)}$ and  \begin{gather}\label{reverse-time}
 \varphi_z[\rho](-t)=\pi^{-1/2}e^{it(w,z)-\frac{1}{2}t^2 (z,Sz)}\,.
 \end{gather}
Thereby, \eqref{charact-Gaussian} and \eqref{reverse-time} imply $e^{-\frac12 t^2(z,Sz)}=\tr{\rho e^{it\lrp{p(z)-(w,z)}}}$ and by \eqref{moments-derivatives}, 
\begin{align*}
\wm{\lrp{p(z)-(w,z)}^n}_\rho= (-i)^n \frac{d^n}{dt^n}\tr{\rho e^{it\lrp{p(z)-(w,z)}}}\rE{t=0}= (-i)^n \frac{d^n}{dt^n}{e^{-\frac12 t^2(z,Sz)}}\rE{t=0}\,,
\end{align*}
which yields the following result.
\begin{corollary}\label{necessity} The $\rho(w,S)$-$n^{\rm th}$ Weyl moments
 satisfy the following recurrence relation  
\begin{align}\label{g-moments}
\wm{\lrp{p(z)-(w,z)}^n}_\rho=\begin{cases}
0\,,&\mbox{for $n$ odd}\\
(z,Sz)^\frac{n}{2}(n-1)!!\,,&\mbox{for $n$ even}
\end{cases}\,.
\end{align}
\end{corollary}

\begin{remark} One can disentangle $e^{it(p(z)-(w,z))}$ to obtain, after taking derivatives,
\begin{gather*}
\wm{\lrp{p(z)-(w,z)}^n}_\rho= \sum_{k=0}^{n}(-1)^k \begin{pmatrix} n \\ k\end{pmatrix}\wm{p(z)^{n-k}}_\rho (w,z)^k\,.
\end{gather*}
The first two moments of $\{m_n(z)=\wm{p(z)^n}_\rho\}_{n\in\N}$ imply the other ones, viz. \eqref{g-moments} is explicitly written as follows 
\begin{align*}
m_n(z)=\,\unit_{2{\mathbb N}}(n)(m_2(z)-m_1^2(z))^{\frac{n}{2}}(n-1)!!+\sum_{k=1}^n(-1)^{k+1}\begin{pmatrix} n \\ k\end{pmatrix}m_{n-k}(z)m_1^k(z)\,,\quad n\geq 3\,,
\end{align*}
where $\unit_{2{\mathbb N}}$ denotes the indicator function of the set of even natural numbers.
\end{remark}
\begin{remark}\label{rm:centered-Gs}
Every Gaussian state can be supposed to be centered, viz. with mean value vector equal zero. Indeed, if $\rho(w,S)$ is a Gaussian state then $\hat\rho=W_{-iw/2}\rho W_{iw/2}$ is a centered Gaussian state, due to Remark~\ref{rm:coherent-chanel}.
\end{remark}

The converse assertion of Corollary~\ref{necessity} also holds in the following sense: regarding the \emph{moment-generating} function of an observable $A$ in a state $\rho$ (cf. \cite{MR4516426})
\begin{gather*}
g_{\rho,A}(x)=\sum_{n\geq0}\frac{1}{n!}\wm{A^n}_\rho x^n\,,\quad x\in\R\,,
\end{gather*}
one obtains the following result.

\begin{theorem}
For a given state $\rho$, if all the $\rho\,$-$n^{\rm th}$ moments of $p(z)$ are finite and satisfy \eqref{g-moments}, for some $w\in\ell_2(\N)$ and selfadjoint real operator $S\in\mathcal B(\ell_2(\N))$, then
\begin{gather}\label{eq:gf-p_z}
g_{\rho,p(z)}(x)=e^{(w,z)x+\frac12(z,S,z)x^2}\,,\quad x\in\R\,.
\end{gather}
Moreover, $\rho$ is Gaussian with mean value vector $w$ and covariance $S$.
\end{theorem}
\begin{proof}
By virtue of Remark~\ref{rm:centered-Gs}, we may suppose that $w=0$. Then, for any $x\in\C$,
\begin{align}\label{eq:series-moments}
\begin{split}
\sum_{n\geq0}\frac{1}{n!}\wm{p(z)^n}_\rho x^n&=\sum_{n\geq 0} \frac{1}{(2n)!} (z,Sz)^n (2n-1)!! x^{2n} \\=& \sum_{n\geq 0} \frac{1}{n!}\lrp{\frac{(z,Sz)}{2}}^n x^{2n}=e^{\frac{1}{2}(z,Sz)x^2}\,,
\end{split}
\end{align}
since $(2n-1)!!=(2n)!(2^nn!)^{-1}$, which implies \eqref{eq:gf-p_z}. Besides, by virtue of \eqref{eq:series-moments}, uniform convergence \cite[Th.\,7.17]{MR0385023} (cf. \cite[prove of Th.\,1]{MR4516426})  and \eqref{eq:Weyl-generator},
\begin{align*}
e^{-\frac{1}{2}(z,Sz)}&=\sum_{n\geq0}\frac{1}{n!}\wm{p(z)^n}_\rho(-i)^n=\sum_{n\geq0}\frac{(-1)^n}{n!}
 {\frac{d^n}{dt^n}\tr{\rho e^{it p(z)}}}\upharpoonleft_{t=0}\\
&=\tr{\rho\sum_{n\geq0}\frac{(-i)^n}{n!}p(z)^n}=\tr{\rho W_z}\,,
 \end{align*}
 as required.
\end{proof}

\subsection{Unbounded normal operators integrable with respect to a state}
 We introduce in this subsection the notion of an integrable normal operator with respect to a state. Here, we also provide necessary and sufficient conditions for a possible unbounded positive self-adjoint operator to be integrable with respect to a state $\rho$.

For any state $\rho$ and any bounded positive selfadjoint operator $A$, we have that 
\begin{gather*}
\tr{\rho A} = \tr{\rho^{\frac{1}{2}} A \rho^{\frac{1}{2}}} =\ip{A^{\frac{1}{2}}\rho^{\frac{1}{2}}}{ A^{\frac{1}{2}} \rho^{\frac{1}{2}}}_2 =\ip{ \rho^{\frac{1}{2}}}{A \rho^{\frac{1}{2}}}_2 \,,
\end{gather*}
i.e., any state can represented by unit vectors $\rho^{\frac{1}{2}}$ in the Hilbert space $L_{2}(\mathsf{h})$, where $\mathsf{h}$ is the space given in \eqref{eq:infinite-TP}. Besides, one can carry the observables and semigroups defined on $\mathsf{h}$ to corresponding observables and semigroups acting on the space $L_2(\mathsf{h})$. Actually, we follow \cite{MR363279} in considering a more general left multiplication operator $M_B$ given by 
\begin{gather*}
M_B \rho = B\rho\,,\quad B\in{\mathcal B}(\mathsf{h})
\end{gather*}
which is  bounded  on $L_{2}(\mathsf{h})$, with norm $\no{M_B}=\no{B}$. The isomorphism ${\mathcal V}$ from $L_2(\mathsf{h})$ onto $\mathsf{h}\otimes \mathsf{h}$ is defined by 
\begin{gather*}
\mathcal V\vk u\vb v=u \otimes \theta v\,,
\end{gather*}
which is extended by linearity and continuity to the whole space $\mathsf{h}$, where $\theta$ is any anti-unitary operator on $\mathsf{h}$ such that $\theta^2 =\unit$. One directly computes that
\begin{gather*}
\mathcal VM_B\mathcal V^{-1}u \otimes v= \mathcal V\vk{B u}\vb{\theta v} = \lrp{B\otimes \unit}u\otimes v\,.
\end{gather*} 
Thereby, it follows by linearity and density that 
\begin{gather*}
\mathcal VM_B\mathcal V^{-1}= B\otimes \unit\,.
\end{gather*}
Hence, we can identify $L_{2}(\mathsf{h})$ with $\mathsf{h} \otimes\mathsf{h}$ and consider the operator $B\mapsto B\otimes \unit$ instead of $M_B$. 

If $U(t)$ is a strongly continuous group of normal operators on $\mathsf{h}$, then the corresponding generator $A$ is normal, with associated spectral measure $(E_{\lambda} )_{\lambda\in{\mathbb R}}$ (cf. \cite[Chap.\,13]{MR1157815}) and \begin{gather}\label{eq:SUG-Ut}
U(t)= \int e^{t\lambda} d E_{\lambda}\,.
\end{gather}
The group $U(t)\otimes \unit$ and  $(E_{\lambda}\otimes \unit)$ have corresponding properties, for instance
\begin{gather*}
U(t)\otimes \unit= \int e^{t\lambda} d (E_{\lambda}\otimes \unit) \,.
\end{gather*}
Let $\lrp{{\mathbb U}(t)}_{t\in{\mathbb R}}$ and $\lrp{{\mathbb E}_{\lambda}}_{\lambda\in{\mathbb R}}$ be the corresponding group and spectral family on $L_2(\mathsf{h})$, and
\begin{gather*}
{\mathcal V}{\mathbb U}_{t}{\mathcal V}^{-1} = U(t)\otimes \unit\,,\quad\forall\, t\in{\mathbb R}\,.
\end{gather*}
Consider the corresponding normal generator ${\mathbb A}$ and the representations
\begin{gather}\label{spectral-L2}
{\mathbb A}=\int \lambda d{\mathbb E}_{\lambda}\,,\quad {\mathbb U}_{t}=\int e^{t\lambda} d{\mathbb E}_\lambda\,,
\end{gather} 
whence  $A$ and ${\mathbb A}$ are selfadjoint or positive selfadjoint simultaneously. This happens when \eqref{eq:SUG-Ut} is a strongly continuous group of normal operators on $\mathsf{h}$ and $\lambda=i\beta$, with $\beta\in\R$. The explicit action of ${\mathbb A}\colon\dom\mathbb A\subset  L_2(\mathsf{h})\to L_2(\mathsf{h})$ is (cf. \cite[Lem.\,2]{MR363279})
\begin{gather}\label{explicit-action}
{\mathbb A}\eta = A \eta, \qquad\mbox{for all } \, 
\eta \in {\rm dom \, }({\mathbb A})=\{\eta\in L_{2}(\mathsf{h}) \, : \, A \eta\in L_{2}(\mathsf{h})\}\,.
\end{gather}

It is well-known that any normal operator $\mathbb A$ is decomposed into a linear combination of four positive selfadjoint parts \cite[Chap.\,3]{MR919948}. Namely, by the spectral theorem, $\mathbb A=\int \lambda d\mathbb E_\lambda$ and  
\begin{gather*}
\mathbb A=\int (\re \lambda)_+ d\mathbb E_\lambda-\int (\re \lambda)_- d\mathbb E_\lambda+i\lrp{\int (\im \lambda)_+ d \mathbb E_\lambda-\int (\im \lambda)_- d\mathbb E_\lambda}\,,
\end{gather*}
getting four positive selfadjoint operators having commuting spectral projections
such that $\mathbb A=(\re \mathbb A)_+-(\re \mathbb A)_-+i((\im \mathbb A)_+-(\im \mathbb  A)_-)$.

\begin{definition}\label{def:A-traceable-rho}
 A not necessarily bounded normal operator $A$ is integrable with respect to a state $\rho$ ($\rho\,$-integrable for short) if $\lim_{\epsilon\to0}\tr{\rho[(\re A)_{\pm}]_\epsilon}, \lim_{\epsilon\to0}\tr{\rho[(\im A)_{\pm}]_\epsilon}$ exist. In such a case, we write  
\begin{gather}\label{eq:def-rhoA-traceable}
 \tr{\rho A}\colonequals\lim_{\epsilon\to 0}\tr{\rho\big([(\rm{Re\, } A)_+)]_\epsilon-[(\rm{Re\, } A)_-)]_\epsilon
+i[(\rm{Im\, } A)_+)]_\epsilon-i[(\rm{Im\, } A)_-)]_\epsilon\big)}\,.
\end{gather}
The usual trace is recovered in \eqref{eq:def-rhoA-traceable} when $A$ is bounded \cite[Rm.\,4]{MR4516426}.
\end{definition}
The following was proved in \cite{MR4516426}.
\begin{lemma}\label{lem-positive-traceable-rho}
A positive selfadjoint operator $A$ is integrable with respect to a state $\rho$ if and only if $A^{\frac12}\rho^{\frac12}\in L_2(\mathsf{h})$. 
In such a case, one has that
\begin{gather}\label{eq:trace-rho-postive}
\tr{\rho A}= \int\lambda d\ip{\rho^{\frac12}}{\mathbb E_\lambda\rho^{\frac12}}_2
=\ip{A^{\frac12}\rho^{\frac12}}{A^{\frac12}\rho^{\frac12}}_2\,,
\end{gather} being $\mathbb{E}_\lambda$ the spectral measure in \eqref{spectral-L2}. Besides, for $\rho = \sum_{k\in\mathbb N} \rho_{k} \vk{u_k}\vb{u_k}$, \eqref{eq:trace-rho-postive} turns into
\begin{gather}\label{eq:trace-rho-decomposition}
\tr{\rho A}=\sum_{k\in\mathbb N} \rho_k\no{A^{\frac{1}{2}}u_k}^2\,.
\end{gather}
\end{lemma}
\begin{remark}\label{rm:normal-integrableO} For a state $\rho$, we have the following:
\begin{itemize}
\item[(i)] A normal $\rho\,$-integrable operator $A$ holds
\begin{gather}\label{eq:normal-integrableO}
\tr{\rho A}=\int \lambda d\ip{\rho^{\frac12}}{\mathbb E_\lambda\rho^{\frac12}}_2=\ip{\rho^{\frac12}}{A\rho^{\frac12}}_2\,.
\end{gather}
Indeed, the positive parts $(\re A)_\pm, (\im A)_\pm$ of $A$ are $\rho\,$-integrable and satisfy \eqref{eq:trace-rho-postive}.
\item[(ii)] If $\alpha\in {\mathbb C}$ and $A, B$ are two normal $\rho\,$-integrable operators for which  $\alpha A + B$ is normal and $\rho\,$-integrable, then for the spectral measure $\mathbb E_\lambda$ of $\alpha A+B$, 
\begin{align}\label{eq:normal-integrableO-2}
\begin{split}
\tr{\rho (\alpha A+B)}&=\int \lambda d\ip{\rho^{\frac12}}{\mathbb E_\lambda\rho^{\frac12}}_2=\ip{\rho^{\frac12}}{\alpha A\rho^{\frac12}}_2 + \ip{\rho^{\frac12}}{B\rho^{\frac12}}_2 \\&= \alpha\tr{\rho A} + \tr{\rho B}\,.
\end{split}\end{align}
\end{itemize}
\end{remark}

\begin{theorem}\label{thm:integrable_norm} A normal operator $A$ is integrable with respect to a state $\rho$ if and only if 
\begin{gather}\label{eq:norm-rho-int}
\no{A}_\rho\colonequals \tr{\rho\abs{A}}=\int \abs{\lambda} d\ip{\rho^{\frac12}}{\mathbb E_\lambda\rho^{\frac12}}_2<\infty\,.
\end{gather}
Besides, if $A^n$ is $\rho\,$-integrable, for $n\geq0$, then so is $A^k$, for $k=0,\dots,n$.
\end{theorem}
\begin{proof}
It is simple to verify from \eqref{eq:trace-rho-postive} that the positive parts of $A$ are $\rho\,$-integrable if and only if 
$\abs{A}_\rho=\int \lrp{\abs{\re \lambda}+\abs{\im \lambda} } d\ip{\rho^{\frac12}}{\mathbb E_\lambda\rho^{\frac12}}_2<\infty$. Besides, $\no A_\rho\leq\abs A_\rho\leq 2\no A_\rho$, which implies the first part of the assertion.  Now, if $A^n$ is $\rho\,$-integrable then, since the measure $\ip{\rho^{\frac12}}{\mathbb E_\lambda(\C)\rho^{\frac12}}_2=1$,   the H\"older inequality implies for $k=1,\dots,n$ (the case $k=0$ is simple) 
\begin{gather*}
\no{A^k}_\rho=\lrp{\int \abs{\lambda}^k d\ip{\rho^{\frac12}}{\mathbb E_\lambda\rho^{\frac12}}_2}\leq \lrp{\int \abs{\lambda}^n d\ip{\rho^{\frac12}}{\mathbb E_\lambda\rho^{\frac 12}}_2}^{\frac kn}=\no{A^n}_\rho^{\frac kn}<\infty\,,
\end{gather*}
as required.
\end{proof}

In the last theorem, we used the fact that $A^n$ is normal when $A$ is normal. Actually, all the elements of the commutative unital $*$-algebra $\wm{A}$ generated by $A$ are normal.

\begin{corollary} For a  state $\rho$ and a normal operator $A$, it follows that $A^n$ is $\rho\,$-integrable for all $n\geq0$, if and only if the elements of $\wm{A}$ are $\rho\,$-integrable. In such a case, \eqref{eq:norm-rho-int} is a norm for $\wm{A}$ and
\begin{gather}\label{eq:linearity-integrable}
\tr{\rho(\alpha f(A)+\beta g(A))}=\alpha\tr{\rho f(A)}+\beta\tr{\rho g(A)}\,,\quad f(A),g(A)\in \langle A\rangle \,.
\end{gather}
\end{corollary}
\begin{proof}
If $A^n$ is $\rho\,$-integrable for all $n\geq0$, then so is $A^{*n}$, since $\tr{\rho A^{*n}}=\int{\cc\lambda}^n d\ip{\rho^{\frac12}}{\mathbb E_\lambda\rho^{\frac12}}_2$. In addition, if $f(A),g(A)\in\langle A \rangle$ are $\rho\,$-integrable, then Theorem~\ref{thm:integrable_norm} implies
\begin{align}\label{eq:triangle-inequality}
\begin{split}
\no{\alpha f(A)+\beta g(A)}_\rho&=\int \abs{\alpha f(\lambda)+\beta g(\lambda)} d\ip{\rho^{\frac12}}{\mathbb E_\lambda\rho^{\frac12}}_2\\
&\leq \abs{\alpha}\int \abs{f(\lambda)}d\ip{\rho^{\frac12}}{\mathbb E_\lambda\rho^{\frac12}}_2+\abs{\beta}\int \abs{g(\lambda)}d\ip{\rho^{\frac12}}{\mathbb E_\lambda\rho^{\frac12}}_2\\&=\abs\alpha\no{f(A)}_\rho+\abs\beta\no{g(A)}_\rho<\infty\,,
\end{split}
\end{align}
i.e., $\alpha f(A)+\beta g(A)$ is $\rho\,$-integrable. The converse is straightforward. The properties for which  $\no\cdot_\rho$  is a norm follow by \eqref{eq:trace-rho-decomposition} and \eqref{eq:triangle-inequality}. Besides, \eqref{eq:linearity-integrable} is a consequence of the spectral representation \eqref{eq:normal-integrableO}.
\end{proof}

The proof of the next assertion follows the same lines as the proof of  \cite[Th.\,2]{MR4516426}.

\begin{theorem}\label{thm:An-traceable-derivates}
Let $\rho$ be a state, $A$ be a normal operator and $n\geq0$. If $A^n$ is $\rho\,$-integrable then
\begin{gather}\label{eq:An-integrable-derivates}
\tr{\rho A^n}=\frac{d^n}{dt^n}\tr{\rho e^{tA}}\rE{t=0}\,.
\end{gather}
\end{theorem}

\begin{remark}\label{rm:traceable-moments} For $n\geq 0$ and a Gaussian state $\rho$ such that $p(z)^n$ is $\rho\,$-integrable, it follows by \eqref{eq:An-integrable-derivates} (after replacing $A$ by $ip(z)$) that the $\rho\,$-$n^{\rm th}$ Weyl moment satisfies
\begin{gather*}
\wm{p(z)^n}_\rho =\tr{\rho p(z)^n}\,.
\end{gather*}
\end{remark}

\section{On the covariance operator}\label{Covariance-operator}

 Recall that the real Hilbert space $(\ell_2(\mathbb{N}),(\cdot,\cdot))$ has $\{\delta_j,e_j\colonequals -i\delta_j\}_{j\geq1}$ as canonical basis, where $\{\delta_j\}_{j\geq1}$ is the canonical basis of $(\ell_2(\mathbb{N}),\ip\cdot\cdot)$. Besides, $(z, u)\colonequals\re\ip zu$. Thus, it is decomposable into two real subspaces by
 \begin{gather*}
\ell_2(\mathbb{N})=\mathcal H_\delta\oplus\mathcal H_e\,,\quad\mbox{where}\quad\mathcal H_\delta\colonequals\cc{\Span\lrb{\delta_j}}_{j\geq1}\,,\quad \mathcal H_e\colonequals\cc{\Span\lrb{e_j}}_{j\geq1}\,.
\end{gather*}

 If  $\rho(w,S)$ is a Gaussian state then the mean value vector can be decomposed into $w= \sqrt{2}(l+m)$, where $l\in\mathcal H_\delta$ and $m\in\mathcal H_e$ are called the mean \emph{momentum} and \emph{position} vectors, respectively. Besides, the covariance operator follows the matrix representation $S=[(u_r,Sv_s)]_{r,s\geq1}$, where $u_r,v_s\in\{\delta_j,e_j\}_{j\geq1}$.

Recall that $p(z)p(u)$ is a normal operator, this is a consequence of  \eqref{commutation-sigma}.

\begin{definition}
We call a state $\rho$ \emph{amenable} if  $p(z)p(u)$ is $\rho\,$-integrable, for all $z,u\in \ell_2(\mathbb{N})$.
\end{definition}

\begin{remark} For an amenable Gaussian state $\rho$, it follows that $p(z)^2$ is $\rho\,$-integrable as well as $p(z)$ by Theorem~\ref{thm:integrable_norm}. Thereby,  Remark~\ref{rm:traceable-moments} asserts that  \begin{gather}\label{eq:equivalence-MI}
\wm{p(z)^n}_\rho =\tr{\rho p(z)^n}\,,\quad n=1,2\,,\quad z\in\ell_2(\N)\,.
\end{gather}
Thus, we will use freely both expressions in \eqref{eq:equivalence-MI} for amenable Gaussian states.
\end{remark}

\begin{proposition}\label{w-S-formulae} For amenable Gaussian state $\rho(w, S)$ the following hold:
\begin{enumerate}[(i)]
\item\label{it1:vMcM} The mean value vector satisfies 
\begin{gather}\label{eq:vMcM}
w=\sqrt{2}\sum_{j\geq1}\lrp{\wm{p_j}_\rho\delta_j+\wm{q_j}_\rho e_j}\,.
\end{gather}
Hence, its mean momentum and position vectors are given from \eqref{eq:equivalence-MI} by
\begin{gather*}
l=\sum_{j\geq1}\tr{\rho p_j}\delta_j\quad\mbox{and}\quad m=\sum_{j\geq1}\tr{\rho q_j}e_j\,,\quad\mbox{respectively}\,.
\end{gather*}

\item\label{it2:vMcM} It follows for any $z,u\in\ell_2(\N)$ that 
\begin{gather}\label{eq:covariance-entries}
\displaystyle (z,Su)=\re\tr{\rho p(z)p(u)}-\wm{p(z)}_\rho\wm{p(u)}_\rho\,.
\end{gather}
Therefore, for $j,k\geq1$, the entries of the covariance operator holds
\begin{align*}
(\delta_j,S\delta_k)&=2\lrp{\re \tr{\rho p_jp_k}-\tr{\rho p_j}\tr{\rho p_k}}\,,\\
(e_j,Se_k)&=2\lrp{\re \tr{\rho q_jq_k}-\tr{\rho q_j}\tr{\rho q_k}}\,,\\
(\delta_j,Se_k)&=2\lrp{\re\tr{\rho p_jq_k}-\tr{\rho p_j}\tr{\rho q_k}}\,,
\end{align*}
being $p_j,q_j$ the operators given in \eqref{eq:MP-CA}. In particular, 
\begin{gather*}
(\delta_j,S\delta_j)=2\lrp{\tr{\rho p_j^2}-\lrp{\tr{\rho p_j}}^2}\,;\quad 
 (e_j,Se_j)=2\lrp{\tr{\rho q_j^2}-\lrp{\tr{\rho q_j}}^2}\,.
\end{gather*}
\end{enumerate}
\end{proposition}
\begin{proof}
It follows from \eqref{g-moments}, for $n=1$, that $(w,z)=\wm{p(z)}_\rho$. In this fashion, if $w=\sum_{j\geq1}a_j\delta_j+b_je_j$, with $a_j,b_j\in\R$, and  $z\in\{\delta_j,e_j\}_{j\geq1}$, then by \eqref{eq:MP-CA},
\begin{gather*}
a_j=\sqrt{2}(w,\delta_j/\sqrt{2})=\sqrt{2}\wm{p_j}_\rho\,,\quad
b_j=\sqrt{2}(w,e_j/\sqrt{2})=\sqrt{2}\wm{q_j}_\rho\,,
\end{gather*}
whence it follows \eqref{eq:vMcM}. Now, one computes \eqref{eq:linearity-integrable} and \eqref{g-moments} that $(w,z)=\wm{p(z)}_\rho$ and $(z,Sz)=\wm{p(z)^2}_\rho-\wm{p(z)}_\rho^2$. Then, one gets for $z,u\in\ell_2(\N)$ that
\begin{align}\label{eq:dist-uHvU}
\begin{split}
\lrp{z+u,S(z+u)}&=2(z,Su)+(z,Sz)+(u,Su)\\
&=2(z,Su)+\wm{p(z)^2}_\rho-\wm{p(z)}_\rho^2+\wm{p(u)^2}_\rho-\wm{p(u)}_\rho^2\,.
\end{split}
\end{align}
Besides, one has by \eqref{eq:linear-P-inZ} that
\begin{align*}
\lrp{z+u,S(z+u)}&=\wm{p(z+u)^2}_\rho-\wm{p(z+u)}_\rho^2=\wm{(p(z)+p(u))^2}_\rho-(w,z+u)^2
\\&=\wm{(p(z)+p(u))^2}_\rho-\wm{p(z)}_\rho^2-\wm{p(u)}_\rho^2-2\wm{p(z)}_\rho\wm{p(u)}_\rho\,,
\end{align*}
which comparing with  \eqref{eq:dist-uHvU}, 
\begin{gather}\label{eq:aux1-wm}
2(z,Su)=\wm{(p(z)+p(u))^2}_\rho-\wm{p(z)^2}_\rho-\wm{p(v)^2}_\rho-2\wm{p(z)}_\rho\wm{p(u)}_\rho\,.
\end{gather}
On the other hand, due to \eqref{eq:trace-rho-postive}, \eqref{eq:normal-integrableO} and since $\rho$ is amenable, 
\begin{align*}
\wm{(p(z)+p(u))^2}_\rho&=\no{(p(z)+p(u))\rho^{1/2}}_2^2=
\wm{p(z)^2}_\rho+\wm{p(u)^2}_\rho+2\re\ip{p(z)\rho^{1/2}}{p(u)\rho^{1/2}}_2\\
&=\wm{p(z)^2}_\rho+\wm{p(u)^2}_\rho+2\re\tr{\rho p(z)p(u)}\,,
\end{align*}
which replacing in \eqref{eq:aux1-wm} one arrives at \eqref{eq:covariance-entries}.
\end{proof}
Let us denote the \emph{variance} of an observable $A$ in the state $\rho$ by 
\begin{gather*}
V_\rho(A)^2\colonequals\tr{\rho A^2} - \tr{\rho A}^2\,,
\end{gather*}
whenever it makes sense. 
\begin{remark}
For an amenable Gaussian state $\rho(w,S)$, it follows from \eqref{eq:linearity-integrable} and \eqref{eq:trace-rho-postive} that
\begin{gather}\label{eq:variance-norm2}
V_{\rho}\lrp{p(z)}^2=\no{\lrp{p(z)-\wm{p(z)}_\rho } \rho^{\frac{1}{2}}}_{2}^{2} \,, \quad z\in\ell_2(\mathbb{N})\,.
\end{gather} 
\end{remark}

The following inequality yields the so-called \emph{Heisenberg's uncertainty principle}. 

\begin{theorem} For an amenable Gaussian state $\rho(w,S)$, the following holds for $z,u\in\ell_2(\C)$:
\begin{gather}\label{eq:uncertainty-principle}
V_{\rho}\lrp{p(z)}^2V_{\rho}\lrp{p(u)}^2\geq (z,Su)^2+\abs{\im\ip zu}^2\,.
\end{gather}
Hence, the uncertainty principle $V_{\rho}\lrp{p(z)}^2V_{\rho}\lrp{p(u)}^2\geq \abs{\im\ip zu}^2$ holds true.
\end{theorem}
\begin{proof}
It is easy to check from \eqref{commutation-sigma} that $\im \tr{\rho p(z)p(u)}=\im\ip{z}{u}$ and by \eqref{eq:covariance-entries},
\begin{align}\label{eq:aux-by-cSi}
\begin{split}
\big\langle\lrp{p(z)-\wm{p(z)}_\rho } \rho^{\frac{1}{2}},\lrp{p(u)-\wm{p(u)}_\rho } \rho^{\frac{1}{2}}\big\rangle&=\tr{\rho p(z)p(u)}-\wm{p(z)}_\rho\wm{p(u)}_\rho\\
&=(z,Su)+i\im\ip zu\,.
\end{split}
\end{align}
Therefore, we get by Schwartz inequality, \eqref{eq:variance-norm2} and \eqref{eq:aux-by-cSi} that 
\begin{align*}
V_{\rho}\lrp{p(z)}^2V_{\rho}\lrp{p(u)}^2&=\no{\lrp{p(z)-\wm{p(z)}_\rho } \rho^{\frac{1}{2}}}_{2}^{2}\no{\lrp{p(u)-\wm{p(u)}_\rho } \rho^{\frac{1}{2}}}_{2}^{2}\\
&\geq\abs{\big\langle\lrp{p(z)-\wm{p(z)}_\rho } \rho^{\frac{1}{2}},\lrp{p(u)-\wm{p(u)}_\rho } \rho^{\frac{1}{2}}\big\rangle}^2\\
&=(s,Su)^2+\abs{\im\ip zu}^2,
\end{align*}
as required.
\end{proof}

For an amenable Gaussian state $\rho(w,S)$, one has by virtue of \eqref{eq:covariance-entries} that
\begin{gather}\label{eq:positive-CM}
(z,Sz)=\wm{p(z)^2}_\rho-\wm{p(z)}_\rho^2=V_\rho(p(z))^2\,,\quad z\in\ell_2(\mathbb{N})\,.
\end{gather}
In this fashion, \eqref{eq:uncertainty-principle} implies 
\begin{gather}\label{eq:MatrixS-inequality}
(z,Sz)(u,Su)\geq (z,Su)^2+\abs{\im\ip{z}{u}}^2\,,\quad\mbox{for all}\quad z,u\in\ell_2(\N)\,.
\end{gather}

\begin{corollary}\label{cor:uncertainty-principle}
If $\rho(w, S)$ is an amenable Gaussian state, then the cova\-riance operator $S\in{\mathcal B}_{\mathbb R}({\ell_2(\mathbb{N})})$ is positive and invertible, with $S^{-1}\in{\mathcal B}_{\mathbb R}({\ell_2(\mathbb{N})})$. Besides, 
\begin{gather}\label{eq:J-condition}
S - iJ\geq 0
\end{gather}
where $J$ is the operator of multiplication by $-i$ on $\ell_2(\mathbb{N})$.
\end{corollary}
\begin{proof}
The operator $S$ is positive due to \eqref{eq:variance-norm2}  and \eqref{eq:positive-CM}. Besides, if $Sz_0=0$ then $V_\rho(p(z_0))^2=0$ and \eqref{eq:variance-norm2} implies  that $p(z_0)=\tr{\rho p(z_0)}I$, i.e., $z_0=0$, since $p(z)$ is unbounded for $z\neq0$. Thereby, $\ker S=\{0\}$ and $\cc{\ran S}=\ell_2(\mathbb{N})$. Now, $\no{Sz}\leq\no S\no z$, for $z\in\ell_2(\N)$, since $S$ is bounded. Thus, taking into account \eqref{eq:MatrixS-inequality} and the Schwartz inequality, 
\begin{align}\label{eq:qux-J-condition}
\begin{split}
\no z^4&=\abs{\im\ip {iz}z}^2\leq (iz,Siz)(z,Sz)\\&\leq\no {iz}\no{Siz}\no{z}\no{Sz}\leq \no{z}^3\no S\no{Sz}\,,
\end{split}
\end{align}
which implies $\no z\leq \no S\no {Sz}$. Thereby, letting $z=S^{-1}u$, for $u\in\ran S$, one arrives at $\no {S^{-1}u}\leq \no S\no {u}$, i.e., $S^{-1}$ is bounded and $\no{S^{-1}}\leq\no S$, viz. $S^{-1}\in{\mathcal B}_{\mathbb R}({\ell_2(\mathbb{N})})$, since $S$ and $S^{-1}$ are closed simultaneously. Observe that 
$(iz,Siz)=\re(-i\ip{z}{Siz})=\re\ip{S^*z}{z}=(z,Sz)$ and $(z,z)=(z,iJz)$. In this fashion, \eqref{eq:qux-J-condition} produces $(z,iJz)=\no z^2\leq (z,Sz)$, i.e., $0\leq (z,(S-iJ)z)$, whence it follows \eqref{eq:J-condition}.
\end{proof}

\begin{remark}
The covariance operator of an amenable Gaussian state $\rho(w,S)$ satisfies
\begin{gather*}
0<\|{S^{-1}}\|\leq1\leq\no S<\infty\,.
\end{gather*}
Indeed, Corollary~\ref{cor:uncertainty-principle} implies that $S,S^{-1}\in{\mathcal B}_{\mathbb R}({\ell_2(\mathbb{N})})$, they are positive and $\no z^2\leq(z,Sz)\leq \no z^2\no S$, i.e., $\no S\geq1$. In addition, for $z=S^{-1}u$, it follows that $\no{S^{-1}u}^2\leq (S^{-1}u,u)\leq \no u^2\no{S^{-1}}$, viz. $\no {S^{-1}}\leq1$, as expected.
\end{remark}

\subsection*{Acknowledgment}
This work was partially supported by CONACYT-Mexico Grants CBF2023-2024-1842, CF-2019-684340 and UAM-DAI 2024: ``Enfoque Anal\'itico-Combinatorio y su Equivalencia de Estados Gaussianos". The authors are grateful to the anonymous referee for the attentive reading, comments and recommendations on the manuscript.


\def\cprime{$'$} \def\lfhook#1{\setbox0=\hbox{#1}{\ooalign{\hidewidth
  \lower1.5ex\hbox{'}\hidewidth\crcr\unhbox0}}} \def\cprime{$'$}
  \def\cprime{$'$} \def\cprime{$'$} \def\cprime{$'$} \def\cprime{$'$}
  \def\cprime{$'$} \def\cprime{$'$}
\providecommand{\bysame}{\leavevmode\hbox to3em{\hrulefill}\thinspace}
\providecommand{\MR}{\relax\ifhmode\unskip\space\fi MR }
\providecommand{\MRhref}[2]{%
  \href{http://www.ams.org/mathscinet-getitem?mr=#1}{#2}
}
\providecommand{\href}[2]{#2}

\end{document}